\documentclass[aps,twocolumn,superscriptaddress]{revtex4-1}

\usepackage{graphicx}
\usepackage{dcolumn}
\usepackage{bm}
\usepackage{braket}
\usepackage{dsfont}
\usepackage{xcolor}

\usepackage{hyperref}
\usepackage{amsmath}
\usepackage{amsfonts}
\usepackage{amsthm}
\usepackage{enumitem, hyperref}
\usepackage{mathtools}
   \newtheorem{theorem}{Theorem}

   \newtheorem{corollary}{Corollary}
   \newtheorem{lemma}{Lemma}

   \newtheorem{definition}{Definition}

\begin{document}


\title{Reconstruction of Quantum Fields: CCR, CAR and Transfields}

\author{Nicolás Medina S\'anchez}
\email{nicolas.medina.sanchez@univie.ac.at}
\affiliation{University of Vienna, Faculty of Physics, Vienna Center for Quantum Science and Technology (VCQ), Boltzmanngasse 5, 1090 Vienna, Austria}
\affiliation{University of Vienna, Vienna Doctoral School in Physics, Boltzmanngasse 5, 1090 Vienna, Austria}
\author{Borivoje Daki\'c}
\email{borivoje.dakic@univie.ac.at}
\affiliation{University of Vienna, Faculty of Physics, Vienna Center for Quantum Science and Technology (VCQ), Boltzmanngasse 5, 1090 Vienna, Austria}
\affiliation{Institute for Quantum Optics and Quantum Information (IQOQI), Austrian Academy of Sciences, Boltzmanngasse 3, 1090 Vienna, Austria.}

\date{\today}

\begin{abstract}

One of the traditional ways of introducing bosons and fermions is through creation--annihilation algebras. Historically, these have been associated with emission and absorption processes at the quantum level and are characteristic of the language of second quantization. In this work, we formulate the transition from first to second quantization by taking quotients of the state spaces of distinguishable particles, so that the resulting equivalence classes identify states that contain no information capable of distinguishing between particles, thereby generalising the usual symmetrisation procedure. Assuming that the resulting indistinguishable-particle space (i) admits an ordered basis compatible with how an observer may label the accessible modes, (ii) is invariant under unitary transformations of those modes, and (iii) supports particle counting as a mode-wise local operation, we derive a new class of creation–annihilation algebras. These algebras reproduce the partition functions of \emph{transtatistics}, the maximal generalisations of bosons and fermions consistent with these operational principles.

\end{abstract}

\maketitle
\section{Introduction}
A standard textbook presentation of the transition from the quantum theory of a single particle to that of many particles proceeds by first forming the Hilbert space
\[
T(H)=\bigoplus_{n\ge 0} H^{\otimes n},
\]
which describes distinguishable particles whose single--particle spaces are isomorphic. Indistinguishability is then imposed by projecting onto the symmetric or antisymmetric subspaces, yielding the usual occupation--number description in which only the number of particles per accessible mode is relevant. In this traditional formalism, the exclusivity of Bose and Fermi statistics is assumed rather than derived, and any alternative statistical behaviour is excluded \emph{a priori}. Nonetheless, a variety of classical and recent works have exhibited consistent generalisations of statistics, including parastatistics \cite{Green53,Greenberg1964,HARTLE69,ryan1963,stolt1970,stoilova,WangHazzardProvided}, quons \cite{Greenberg1991}, Gentile statistics \cite{Gentile40}, and several topological or combinatorial extensions \cite{dewitt,myrheim,Wilczek82,Read_2003,fredenhagen89,Surya_2004,Balachandran_2001,Baez06,vanenk2019,ZhangHanssonKivelson1989,Witten1989,Haldane08,Laughlin1983}. These are typically introduced by \emph{ad hoc} modifications of commutation relations, rather than through a structural analysis of indistinguishability.

In this work we develop a systematic treatment of the passage from distinguishable to indistinguishable particles based on operational constraints. We again take as starting point the algebra $T(H)$ whose elements represent labelled multi--particle states, following Segal's tensor-algebra formulation of second quantization \cite{Segal1956a,Segal1956b,Segal1958}. Indistinguishability is implemented by an equivalence relation that removes precisely the information operationally inaccessible for distinguishing one constituent from another. The quotient
\[
\mathcal{F} := T(H)/I
\]
then plays the role of a Fock space, where the ideal $I$ encodes which labelled states are identified and thus determines the algebraic rules obeyed by creation and annihilation operators.

To connect this with experimental structure, we assume a factorization
\[
H \cong E \otimes K,
\]
where $E$ comprises the degrees of freedom accessible to the experimenter (for example the external modes), while $K$ contains internal or hidden degrees of freedom as in treatments of indistinguishable particles with internal structure \cite{Benatti2020,Morris2020,NeoriGoyal2012,Neori16,MedinaSanchezDakic2024}. On this decomposition we impose three requirements:

\begin{enumerate}
\item \textbf{Homogeneity.}  
The ideal $I$ respects the natural grading of $T(H)$, so that $\mathcal{F}$ retains a well defined particle--number decomposition.

\item \textbf{Ordered--monomial basis.}  
Fixing an ordering on a basis of $E$, the ordered monomials
\[
\psi_{i_1} \otimes \cdots \otimes \psi_{i_n}, \qquad i_1 \le \cdots \le i_n,
\]
descend to a basis of $\mathcal{F}$. This expresses that the experimenter may label accessible modes arbitrarily and that these labels induce an ordering compatible with indistinguishability.

\item \textbf{Unitary invariance.}  
The ideal $I$ is invariant under the natural action of $U(E)$ on the accessible degrees of freedom, ensuring that experimentally realisable unitary transformations act consistently on equivalence classes.
\end{enumerate}

Condition~2 has a strong algebraic consequence: it forces $I$ to be \emph{quadratically generated}. Thus all relations needed to describe indistinguishability arise already at the level of two--particle states, placing our construction within the general theory of quadratic and Koszul algebras \cite{Priddy1970,PolishchukPositselski2005,BeilinsonGinzburgSoergel1996,LodayVallette2012,Davydov}. This mirrors the bosonic and fermionic cases and contrasts with other generalised statistics whose defining relations may be cubic or of higher order \cite{Green53,Greenberg1964,Greenberg1991,Greenberg90,Fivel90,CHEN_1996}.

Condition~3 then determines which quadratic relations are admissible. Since the action of $U(E)$ is carried entirely by the accessible factor, the invariant quadratic subspaces of $H \otimes H$ are precisely
\[
\mathrm{Sym}^{2}(E)\otimes W_{\mathrm{sym}}
\;\oplus\;
\wedge^{2}(E)\otimes W_{\mathrm{ext}},
\]
for subspaces $W_{\mathrm{sym}}, W_{\mathrm{ext}} \subset K \otimes K$. These internal subspaces are realised as images of projectors on $K \otimes K$, and the existence of an ordered basis (what we called the Poincar\'e--Birkhoff--Witt property) imposes Yang--Baxter type identities on those projectors, connecting our construction to $R$-matrix and reflection-equation algebras \cite{Yang1967,Drinfeld1988,Jimbo1985,Baxter1972,FaddeevReshetikhinTakhtajan1990,Gurevich1991,GurevichPyatovSaponov1996,Majid1993,Majid2000,BM2000,Davydov,BozejkoKummererSpeicher1997}

Within this class of quadratic algebras, which we call PBW \(U(d)\)-equivariant algebras, we obtain a complete characterisation of the admissible single-mode partition functions. Our main result (Theorem~2 or \emph{Field Theorem}) states that a formal power series \(G(t) \in \mathbb{Z}_{\ge 0}[[t]]\) arises as the single-mode Hilbert–Poincaré series of such an algebra if and only if it is a rational function of the form
\[
G(t) = \frac{Q_-(t)}{Q_+(t)},
\]
where \(Q_\pm(t) \in \mathbb{Z}[t]\), \(Q_+(0)=1\), and all roots of \(Q_+\) (resp. \(Q_-\)) are real and strictly positive (resp. strictly negative). This coincides precisely with the class of partition functions identified by the operational \emph{Partition Theorem} developed in \cite{MedinaSanchezDakic2024}.

The quadratic data determine the pure creation and pure annihilation algebras. To obtain their mixed relations and thus a full creation--annihilation algebra, one introduces an exchange map
\[
C \colon H^{\ast} \otimes H \longrightarrow H \otimes H^{\ast},
\]
whose form is fixed by the creation and annihilation projectors together with the vacuum two--point function (equivalently, the inner product on $H$). The resulting algebras admit natural representations of $\mathfrak{gl}(d)$, providing a complete quantum mechanical framework based on these generalised structures.

The paper is organised as follows. In Section~\ref{sec:from_first_to_second} we formulate the transition from first to second quantization as a quotient of the tensor algebra and recall Segal's theorem in this language. Section~\ref{sec:quadratic_spaces} introduces the quadratic quotients associated with generalised statistics, proves that order implies quadratic generation, and classifies \(\mathrm{U}(d)\)--invariant quadratic subspaces. We also derive Yang--Baxter constraints and describe the computation of the single--mode partition function. Section~\ref{sec:creation_annihilation} develops the corresponding creation and annihilation algebras and their mixed sector, and constructs the transfield representation of \(\mathfrak{gl}(d)\). In Section~\ref{sec:example} we work out a concrete finite--order example. Section~\ref{sec:outlook} concludes with some remarks and open problems. All proofs are collected in the appendix.

\section{From first to second quantization as a quotient}\label{sec:from_first_to_second}

\subsection{Operational indistinguishability and information loss}

Consider a family of systems that, at the single--particle level, are described by the same Hilbert space \(H\). An individual particle state is represented by a unit vector \(\psi\in H\), and a collection of \(n\) labelled particles is described by a vector in \(H^{\otimes n}\). The labels encode individuality: the factor \(H\) in each tensor position corresponds to a specific system, and permutations of the factors implement relabellings.

If we now assume that the particles are operationally indistinguishable, then no experiment is able to track which label occupies which position in the tensor product. In particular, any two states whose reduced density matrices on the algebra of observables coincide for all experiments must be identified. This leads to an equivalence relation on \(T(H)\), the space of states for distinguishable particles. In that sense, the many--particle Hilbert space is then defined as the quotient
\[
F \;=\; T(H)\big/ I,
\]
where \(I\subset T(H)\) is a two--sided ideal capturing the information that is lost when passing from labelled to indistinguishable particles.

In the standard approach, indistinguishability is represented by the action of the symmetric group \(S_{n}\) on \(H^{\otimes n}\), and the subspace of physical states is obtained by projecting onto the totally symmetric or totally antisymmetric subspace. This can be rephrased as choosing \(I\) to be the ideal generated by the antisymmetric or symmetric components in degree two, respectively. The equivalence relation introduced in that case is the one that corresponds to the identification of states that only differ on a permutation of the labels. Depending on how the group $S_n$ is acting on the states, we will get either the symmetric or antisymmetric states. 

\subsection{Fock space as a quotient algebra:}

The tensor algebra
\[
T(H) \;=\; \bigoplus_{n\geq 0} H^{\otimes n}
\]
carries a natural \(\mathbb{N}_{0}\)--grading by tensor degree, which determines the number of particles. We assume that the ideal \(I\) is homogeneous,
\[
I \;=\; \bigoplus_{n\geq 0} I_{n}, \qquad I_{n} \subset H^{\otimes n},
\]
so that the quotient inherits a grading
\[
F \;=\; \bigoplus_{n\geq 0} F_{n}, \qquad F_{n} \;=\; H^{\otimes n} / I_{n}.
\]
In this way the transition from distinguishable to indistinguishable particles does not change the number of particles, as expected. Additionally we see that the degree--zero component \(F_{0}\) is one--dimensional, spanned by the vacuum \(|0\rangle\), and the degree--one component \(F_{1}\) is naturally isomorphic to the one-particle space \(H\).

To construct creation operators we choose a specific basis \(\{X_{a}\}_{a\in A}\) of \(H\), where \(A\) is a finite index set, and we denote by the same symbols the elements of degree one in \(T(H)\) given that $H$ descends to $F_1$ after the quotient. We can then build the algebra of all linear combinations of words written in terms of the objects \(\{X_{a}\}_{a\in A}\), or differently, the free unital algebra on the alphabet \(\{X_{a}\}\) where multiplication is given by concatenation of words, that we know is isomorphic to \(T(H)\). \\
A basis element of this algebra then can be written for example as $X_3\,X_1^3\,X_2\,X_7^9\,X_3$ where we have omitted the tensor product between the terms for convenience. In this way we can define creation operators simply using left multiplication on $T(H)$ and then projecting using the equivalence relations:
\[
\hat{X}_{a}^{\dagger}\colon F \to F, \qquad \hat{X}_{a}^{\dagger} [w] \;=\; [X_{a} w],
\]
where \(w\in T(H)\) and \([\cdot]\) denotes the projection to the quotient \(F\). Any many--particle state can be written as a linear combination of classes of words in the alphabet \(\{X_{a}\}\) acting on the vacuum:
\[
[w] \;=\; \hat{X}_{a_{1}}^{\dagger}\cdots \hat{X}_{a_{n}}^{\dagger} |0\rangle.
\]
The structure of the quotient is thus encoded in the relations among these words, equivalently in the form of the ideal \(I\).

\subsection{Segal's theorem}
The construction we employ was used by Segal, who asked under which structural conditions the quotient procedure leads inevitably to the familiar bosonic or fermionic algebras. His
starting point was to impose stability requirements on the ideal that generates the
equivalence relation.

\begin{definition}[Characteristic and fully characteristic ideals]
Let $\mathrm{Aut}_0\bigl(T(H)\bigr)$ denote the subgroup of algebra automorphisms of
$T(H)$ induced by number-preserving unitaries on $H$. An ideal $I \subset T(H)$ is
\emph{characteristic} if
\[
\varphi(I) = I
\qquad\text{for all }\varphi \in \mathrm{Aut}_0\bigl(T(H)\bigr),
\]
and \emph{fully characteristic} if
\[
\psi(I) = I,
\,\forall\, \psi:T(H)\rightarrow T(H)\text{ commuting with }\mathrm{Aut}_0\bigl(T(H)\bigr).
\]
\end{definition}

Segal then considered ideals that remain invariant not only under these automorphisms
but also under the action of the creation operators themselves, and that are maximal
with respect to these invariance properties. The result is the following characterisation.

\begin{corollary}[Segal]\label{cor:Segal}
Let $\mathcal{V}$ be a subspace of the algebra of covariant tensors over a
complex Hilbert space $H$ that is
\begin{enumerate}
    \item characteristic;
    \item invariant under the creation operators $C(x,A)$;
\end{enumerate}
and assume that $\mathcal{V}$ is maximal among subspaces with these properties
that do \emph{not} contain all tensors from some rank onward.
Then $\mathcal{V}$ is either the symmetric ideal or the skew--symmetric ideal.
\end{corollary}

Segal's theorem draws a sharp boundary: if one insists on maximality and full invariance of the ideal under the creation operators themselves, the only possible statistics are Bose and Fermi. This is a powerful rigidity result, but its strength comes precisely from the stringency of its hypotheses. In particular, the maximality condition privileges the largest possible identification of states, collapsing the ideal to one of two extremes. If instead one is willing to consider ideals that are invariant under mode transformations but not necessarily maximal, a much richer landscape of consistent statistics opens up.\\
The question then becomes structural: among all homogeneous ideals compatible with the operational requirements laid out in the introduction, which ones lead to well-defined many-particle spaces with a sensible physical interpretation? The answer, as we now show, is that the three conditions imposed on the Fock space (a graded particle-number decomposition, an ordered monomial basis, and invariance under unitary mode transformations) are not only physically natural but also algebraically decisive. Together they force the ideal to be quadratically generated, placing the entire construction within the well-developed framework of quadratic and Koszul algebras

\section{Quadratic spaces and generalised statistics}\label{sec:quadratic_spaces}

We now introduce the algebraic framework that realises certain transtatistics~\cite{MedinaSanchezDakic2024} as quadratic quotients of the tensor algebra.

\subsection{Internal degeneracies and quantum grammar}
We begin by fixing a factorization \(H\cong \mathbb{C}^{d}\otimes K\) of the 1-particle space, with $E$ the external degrees of freedom and \(K\) certain internal degeneracy space. Let \(\{e_{i}\}_{i=1}^{d}\) be the standard basis of \(\mathbb{C}^{d}\) and let \(\{f_{\alpha}\}_{\alpha=1}^m\) be a basis of \(K\). We define in this way generators
\[
X_{i\alpha} \;:=\; e_{i}\otimes f_{\alpha} \in H.
\]
Along the lines of what we presented in the introduction, we set \(\{X_{i\alpha}\}\) as the alphabet that will be associated with elementary creation operators. Words in this alphabet represent labelled many--particle configurations, and relations among these words define a ``quantum grammar'' for the statistics: they specify which sequences of creation operators are admissible and which are identified or annihilated once the quotient is performed.

The external index \(i\) labels the external modes and transforms by a representation of \(\mathrm{U}(d)\), while the internal index \(\alpha\) labels hidden degrees of freedom and transforms in a way we do not control in principle. We will not fix the latter action henceforth. Instead, it will be shaped by the quadratic relations.
\medskip

The restriction to quadratic relations admits a direct conceptual justification in terms of indistinguishability. In the standard formulation, indistinguishability is implemented by requiring invariance of the state under the action of the symmetric group, whose elements are generated by pairwise transpositions. Thus, the indistinguishability of an arbitrary number of particles is completely determined by the consistency of exchanges at the level of particle pairs. In the present framework, this principle is reflected algebraically by requiring that all relations encoding indistinguishability arise from constraints in the two-particle sector. Higher-degree relations would introduce additional structure not dictated by pairwise exchange, and would therefore go beyond the minimal requirement of indistinguishability. In this sense, quadratic generation captures precisely the idea that indistinguishability is a pairwise notion, and ensures that the extension to many-particle states is governed entirely by consistent two-body exchanges.

\subsection{Order implies quadratic generation}

We now state and use the key structural result mentioned in the introduction.

\begin{lemma}\label{lem:quadratic}
Let \(T = \bigoplus_{n\geq 0} T_{n}\) be a graded free algebra on finitely many generators, and let \(I\triangleleft T\) be a homogeneous two--sided ideal. Suppose that, for a fixed total order on the generators, the classes of ordered monomials form a homogeneous basis of the quotient \(F = T/I\). Then \(I\) is generated in degree two.
\end{lemma}

The intuition behind the proof follows a physical line of reasoning. The transition from distinguishable to indistinguishable particles can be viewed as the introduction of relations among the various polynomials one can form from the chosen reference basis, which plays the role of an alphabet. Imposing that the many–particle space admits an ordered basis means precisely that, after applying the allowed relations, any monomial that does not respect the fixed order can be rewritten as a linear combination of ordered monomials.

The grading by particle number plays an essential role. In degree zero there can be no relations, trivially. In degree one, any relation would collapse information about a single particle, which we do not allow. Thus the minimal degree in which nontrivial relations may appear is the two–particle sector, and they must appear there, since already at two particles one can consider interferometric or scattering scenarios.

Suppose now one were to introduce relations among monomials of degree three, that is, relations allowing any cubic monomial to be rewritten in terms of ordered cubic monomials. Such relations must be consequences of the quadratic ones; otherwise, the degree–three relations could generate unintended dependencies among the ordered quadratic monomials. This would contradict our assumption that the ordered monomials of degree two form a basis. For the statistics to be well defined, every relation involving three or more indistinguishable particles must therefore decompose into relations among pairs of particles. In other words, indistinguishability at the level of $N$ particles must factor through indistinguishability of particle pairs.

Applied to the tensor algebra \(T(H)\) with the generators \(X_{i\alpha}\), Lemma~\ref{lem:quadratic} implies that admissible statistics in our sense are always realised by quadratic quotients. Hence there exists a subspace \(R\subset H\otimes H\) such that
\[
I \;=\; \langle R\rangle,
\]
the ideal generated by \(R\), and \(F\cong T(H)/\langle R\rangle\) is a quadratic algebra.

\subsection{\texorpdfstring{\(\mathrm{U}(1)\)}{U(1)} grading}

The tensor algebra carries a natural \(\mathrm{U}(1)\) action generated by the counting of particles: a phase \(z\in\mathrm{U}(1)\) acts on \(H\) by \(X_{i\alpha}\mapsto z X_{i\alpha}\) and extends multiplicatively to \(T(H)\). We will say that a homogeneous ideal \(I\) is \(\mathrm{U}(1)\)--graded if it is invariant under this action. In that case the quotient \(F\) is also \(\mathrm{U}(1)\)--graded, and the degree \(n\) component transforms with weight \(z^{n}\).

This grading expresses the compatibility of the statistics with local phase transformations that act diagonally in particle number. It is automatically satisfied for homogeneous quadratic ideals, so we take it as part of the basic setup rather than a separate constraint. The important thing is that the assumption of an ordered basis implies that the quotient space carries a canonical decomposition.
\begin{lemma}\label{lem:factorization}
Let $H$ be a complex vector space with a distinguished decomposition
\[
H \;=\; \bigoplus_{i=1}^{d} H_i,
\]
and let $T(H)$ be its tensor algebra. Let $I \triangleleft T(H)$ be a homogeneous two–sided ideal, and set
\[
F \;:=\; T(H)/I.
\]
Assume that the classes of totally ordered monomials with respect to some order in the $X_{i\alpha}$ form a basis of $F$.

Then there exists a vector space $W$ and a linear isomorphism
\[
F \;\cong\; W^{\otimes d}.
\]
Moreover, if $T(H)$ and $I$ are $\mathrm{U}(1)$–graded (by total tensor degree), this is an isomorphism of $\mathrm{U}(1)$–representations, and $W$ is a (finite or countable) reducible $\mathrm{U}(1)$–representation.
\end{lemma}

\subsection{\texorpdfstring{\(\mathrm{U}(d)\)}{U(d)} invariance}

The first non--trivial constraint comes from covariance under external mode transformations. For this to hold for every state, the whole space of states should carry a representation of the unitary group on the external modes $U(E)$. For this to be true the ideal itself and especially the generating subspace should have a a representation of this group as well. The external Hilbert space \(E\cong \mathbb{C}^{d}\) carries the defining representation of \(\mathrm{U}(d)\), and \(H\cong E\otimes K\) inherits an action in which \(\mathrm{U}(d)\) acts on \(E\) and trivially on \(K\). The tensor square decomposes as
\[
H\otimes H \;\cong\; (E\otimes E)\otimes (K\otimes K),
\]
and the \(\mathrm{U}(d)\) action is carried entirely by the factor \(E\otimes E\). This factor decomposes as a direct sum of symmetric and antisymmetric irreducible submodules,
\[
E\otimes E \;\cong\; \mathrm{Sym}^{2}(E)\,\oplus\, \wedge^{2} E.
\]
Thus we have
\[
H\otimes H \;\cong\; \mathrm{Sym}^{2}(E)\otimes (K\otimes K) \;\oplus\; \wedge^{2} E\otimes (K\otimes K).
\]

A \(\mathrm{U}(d)\)--invariant quadratic subspace \(R_{\mathrm{g}}\subset H\otimes H\) must respect this decomposition. By this reasoning we conclude:

\begin{lemma}\label{lem:Udinv}
Let \(H\cong \mathbb{C}^{d}\otimes K\) with the defining action of \(\mathrm{U}(d)\) on \(\mathbb{C}^{d}\). Then any \(\mathrm{U}(d)\)--invariant subspace \(R_{\mathrm{g}}\subset H\otimes H\) is of the form
\[
R_{\mathrm{g}} \;=\; \mathrm{Sym}^{2}(\mathbb{C}^{d})\otimes W_{\mathrm{sym}}\;\oplus\; \wedge^{2}(\mathbb{C}^{d})\otimes W_{\mathrm{ext}},
\]
for some subspaces \(W_{\mathrm{sym}},W_{\mathrm{ext}}\subseteq K\otimes K\).
\end{lemma}

Thus the freedom in choosing \(\mathrm{U}(d)\)--invariant quadratic relations in this construction lies entirely in the choice of the internal subspaces \(W_{\mathrm{sym}}\) and \(W_{\mathrm{ext}}\). These encode how internal degeneracies are coupled to symmetric (bosonic) and antisymmetric (fermionic) external behaviour. The standard bosonic and fermionic cases correspond naturally to the cases where either $W_{\mathrm{sym}}$ or $W_{\mathrm{ext}}$ are isomorphic to $\mathbb{C}$ and the other is the zero subspace.

\subsection{Quadratic relations and the generating ideal}
In this way we can introduce a family of projectors that map into the generating space that determines the ideal. Let \(P^{(d)}_{\mathrm{sym}}\) and \(P^{(d)}_{\mathrm{ext}}\) denote the projectors on \(\mathrm{Sym}^{2}(\mathbb{C}^{d})\) and \(\wedge^{2}(\mathbb{C}^{d})\), and let
\[
P^{(K)}_{\mathrm{sym}}\colon K\otimes K\to W_{\mathrm{sym}},\qquad
P^{(K)}_{\mathrm{ext}}\colon K\otimes K\to W_{\mathrm{ext}}
\]
be projectors onto the chosen internal subspaces, with complementary projectors
\[
P^{(K)\perp}_{\mathrm{sym}} = \mathrm{id}-P^{(K)}_{\mathrm{sym}},\qquad
P^{(K)\perp}_{\mathrm{ext}} = \mathrm{id}-P^{(K)}_{\mathrm{ext}}.
\]
We define
\begin{align*}
P_{\mathrm{g}} &= P^{(d)}_{\mathrm{sym}}\otimes P^{(K)}_{\mathrm{sym} }
                  + P^{(d)}_{\mathrm{ext}}\otimes P^{(K)}_{\mathrm{ext}},\\
P_{\mathrm{g}}^{\perp} &= P^{(d)}_{\mathrm{sym}}\otimes P^{(K)\perp}_{\mathrm{sym} }
                  + P^{(d)}_{\mathrm{ext}}\otimes P^{(K)\perp}_{\mathrm{ext} }.
\end{align*}
These projectors realise the decomposition
\[
H\otimes H \;\cong\; R_{\mathrm{g}} \;\oplus\; R_{\mathrm{g}}^{\perp},
\]
where \(R_{\mathrm{g}} = \mathrm{im}\, P_{\mathrm{g}}\) is the generating quadratic subspace and \(R_{\mathrm{g}}^{\perp}\) is a chosen complement that will be relevant later to the construction of the annihilation operators.

Writing \(X_{i\alpha} = e_{i}\otimes f_{\alpha}\), any element of \(H\otimes H\) can be expressed as
\[
\sum_{i,j,\alpha,\beta} c^{\alpha\beta}_{ij} X_{i\alpha}\otimes X_{j\beta}.
\]
For each pair of external indices \(i,j\) and each elementary internal tensor \(k_{\alpha}\otimes k_{\beta}\in K\otimes K\) we define
\begin{equation}
r^{\alpha\beta}_{ij} := P_{\mathrm{g}}\bigl((e_{i}\otimes e_{j})\otimes (k_{\alpha}\otimes k_{\beta})\bigr).
\end{equation}
An explicit computation shows that
\begin{align}\label{eq:quadratic-relations}
r^{\alpha\beta}_{ij}
 &= \frac{1}{2}\sum_{\gamma,\delta} (K_{\mathrm{sym}})^{\gamma\delta}_{\alpha\beta}
  \bigl(X_{i\gamma}\otimes X_{j\delta} + X_{j\delta}\otimes X_{i\gamma}\bigr)
 \\&+ \frac{1}{2}\sum_{\gamma,\delta} (K_{\mathrm{ext}})^{\gamma\delta}_{\alpha\beta}
  \bigl(X_{i\gamma}\otimes X_{j\delta} - X_{j\delta}\otimes X_{i\gamma}\bigr),
\end{align}
where the tensors \(K_{\mathrm{sym}}\) and \(K_{\mathrm{ext}}\) encode the internal projectors. What we have now is that the generating space identifies precisely those states that will vanish after the quotient. In that sense, these vectors we have just described are the ones that determine the generating relations of the algebra
\[
r^{\alpha\beta}_{ij} = 0, \qquad 1\leq i,j\leq d,\;\; k_{\alpha}\otimes k_{\beta}\in K\otimes K.
\]
We can then denote the corresponding quotient algebra by
\[
A := T(H)/\langle R_{\mathrm{g}}\rangle.
\]

\subsection{Yang--Baxter constraints and PBW property}
All of this construction has been made over the assumption of the existence of an ordered basis after the quotient is performed. However, for the given projections and in consequence the generating relations just introduced to be compatible with this assumption we have to impose what we call the Poincar\'e--Birkhoff--Witt condition (PBW). In the language of quadratic algebras this amounts to requiring that the reduction of words using the quadratic relations is confluent\footnote{A given word in a quadratically generated algebra can be reduced in different ways. The rewriting system given by the defining relations is confluent if no matter which sequence of allowed reductions you follow, you will end up with the same final form}. A convenient formulation uses projectors rather than explicit relations.

Let \(P_{\mathrm{g}}\colon H\otimes H\to R_{\mathrm{g}}\) be the projector onto the generating subspace. Then the following characterisation is standard in the theory of quadratic algebras.

\begin{theorem}\label{thm:YB}
Let \(A = T(H)/\langle R_{\mathrm{g}}\rangle\) be a quadratic algebra defined by a projector \(P_{\mathrm{g}}\) as above. Then \(A\) admits a basis of ordered monomials if and only if, for all \(u,v,w\in H\),
\begin{equation}\label{eq:YB-global}
   (P_{\mathrm{g}}\otimes \mathrm{id})(\mathrm{id}\otimes P_{\mathrm{g}})(P_{\mathrm{g}}\otimes \mathrm{id})
 = (\mathrm{id}\otimes P_{\mathrm{g}})(P_{\mathrm{g}}\otimes \mathrm{id})(\mathrm{id}\otimes P_{\mathrm{g}})
\end{equation}
as endomorphisms of \(H^{\otimes 3}\).
\end{theorem}

The identity \eqref{eq:YB-global} is a quantum Yang--Baxter type equation. When we insert the explicit form of \(P_{\mathrm{g}}\) in terms of symmetric and antisymmetric external projectors and internal projectors \(P^{(K)}_{\mathrm{sym}}\) and \(P^{(K)}_{\mathrm{ext}}\), the condition decouples into two internal Yang--Baxter identities
\begin{align}
P^{(K)}_{\mathrm{sym},12}P^{(K)}_{\mathrm{sym},23}P^{(K)}_{\mathrm{sym},12}
 &= P^{(K)}_{\mathrm{sym},23}P^{(K)}_{\mathrm{sym},12}P^{(K)}_{\mathrm{sym},23},\\
P^{(K)}_{\mathrm{ext},12}P^{(K)}_{\mathrm{ext},23}P^{(K)}_{\mathrm{ext},12}
 &= P^{(K)}_{\mathrm{ext},23}P^{(K)}_{\mathrm{ext},12}P^{(K)}_{\mathrm{ext},23}
\end{align}
on \(K^{\otimes 3}\). These are the precise constraints on the internal data that ensure the PBW property and hence the existence of a well-defined ordered basis. In particular, classical results on quadratic algebras imply that the quotient algebra, i.e.\ the state space of indistinguishable particles, is a Koszul algebra. This structural feature will be used in an essential way later on.

\subsection{Single--mode partition function}

An important feature of the construction that follows from our imposition of the ordered basis is the factorization of the global Fock space into effective single--mode components. Given that our resulting quotient space can be expressed as a tensor power of $U(1)$ representations as
\[
F \;\cong\; W^{\otimes d},
\]
where \(W\) is a graded vector space we have that the Hilbert--Poincar\'e series of \(F\),
\[
H_{F}(t) \;=\; \sum_{n\geq 0} \dim F_{n}\, t^{n},
\]
then factorises as
\[
H_{F}(t) \;=\; G(t)^{d},\qquad
G(t) \;=\; \sum_{n\geq 0} g_{n} t^{n},
\]
where \(G(t)\) is the single--mode series and \(g_{n}=\dim W_{n}\).

The coefficients \(g_{n}\) are determined solely by the quadratic projector \(P_{\mathrm{g}}\) acting on \(K\otimes K\). For each \(n\geq 0\) set \(V_{n}=K^{\otimes n}\) and define operators
\[
P^{(n)}_{k} := \mathrm{id}^{\otimes (k-1)}\otimes P_{\mathrm{g}}\otimes \mathrm{id}^{\otimes (n-k-1)},
\qquad
1\leq k\leq n-1.
\]
The subspace of admissible degree--\(n\) monomials is the intersection
\[
W_{n} \;:=\; \bigcap_{k=1}^{n-1} \ker P^{(n)}_{k},
\]
and the corresponding coefficient is
\[
g_{n} \;=\; \dim W_{n}.
\]
In the language of \cite{MedinaSanchezDakic2024} we know that the Hilbert-Poincar\'e series of the algebra as defined here is in closely related to the partition function, once we exchange the variable $t$ with the Boltzmann factor $e^{\beta\, \epsilon}$. We can then finally state the main theorem:
\begin{theorem}[Field]
Let \(G(t)\in \mathbb Z_{\ge 0}[[t]]\). The following are equivalent:
\begin{enumerate}
\item There exists a homogeneous quadratic quotient
\[
A = T(H)/\langle R\rangle,\qquad H \simeq \mathbb C^d \otimes K,
\]
which is \(U(d)\)-invariant and admits a PBW basis, such that the associated single-mode Hilbert--Poincaré series is \(G(t)\).

\item \(G(t)\) is a rational function of the form
\[
G(t)=\frac{Q_-(t)}{Q_+(t)},
\]
where \(Q_\pm(t)\in \mathbb Z[t]\), \(Q_+(0)=1\), and all roots of \(Q_+\) (resp. \(Q_-\)) are real and strictly positive (resp. strictly negative).
\end{enumerate}
\end{theorem}

This shows that we can generate the partition functions presented in \cite{MedinaSanchezDakic2024} using quadratically generated algebras and more generally that there exists a first quantization picture that can be lifted to a second quantization one.
Interestingly, a classical result due to Priddy shows that algebras admitting a Poincar\'e--Birkhoff--Witt (PBW) basis are Koszul \cite{Priddy1970}. This implies the existence of a precise duality between the Hilbert--Poincar\'e series of the algebra and that of its Koszul dual. Concretely, if
\[
A=\bigoplus_{n\ge 0} A_n , \qquad A_0=\Bbb C,
\]
is a quadratic algebra of the form
\[
A = T(V)/\langle R\rangle ,
\qquad R\subset V\otimes V,
\]
which is Koszul, then its \emph{Koszul dual} $A^!$ is defined as the quadratic algebra
\[
A^! := T(V^*)/\langle R^\perp\rangle ,
\]
where $R^\perp \subset V^*\otimes V^*$ denotes the annihilator of $R$ with respect to the natural pairing
$V^*\otimes V^* \times V\otimes V \to \Bbb C$.
If
\[
H_A(t):=\sum_{n\ge 0} (\dim A_n)\, t^n
\]
denotes the Hilbert--Poincar\'e series of $A$, then the Koszul property implies that the Hilbert--Poincar\'e
series of $A$ and its Koszul dual $A^!$ satisfy the Koszul identity
\begin{equation}\label{eq:koszul-HP}
H_A(t)\, H_{A^!}(-t)=1,
\end{equation}
or equivalently
\[
H_{A^!}(t)=\frac{1}{H_A(-t)}.
\]

In the present setting, applying~\eqref{eq:koszul-HP} at the single--mode level yields
\begin{equation}\label{eq:koszul-single-mode}
G(t)\, G^!(-t)=1,
\quad\text{equivalently}\quad
G^!(t)=\frac{1}{G(-t)}.
\end{equation}

In light of Theorem~2, this means that bosons and fermions, and more generally particles whose single--mode partition function is determined solely by its poles (which we call \emph{transbosons}) and those whose partition function is determined solely by its zeros (which we call \emph{transfermions}), are related by Koszul duality. Indeed, in the transbosonic case $G(t)=1/Q_+(t)$ one finds $G^!(t)=Q_+(-t)$, while in the transfermionic case $G(t)=Q_-(t)$ the dual series is $G^!(t)=1/Q_-(-t)$. This exhibits a genuine algebraic symmetry between these two classes of statistics, implemented by the Koszul duality operation.

In a complementary direction, one has from the canonical theory of PBW algebras that the Hilbert–Poincaré series is necessarily rational. More precisely, by Corollary~6.2 \cite{PolishchukPositselski2005}, for any subset $S \subseteq [1,m]^2$, let $M_S$ denote the $m \times m$ matrix with entries
\[
(M_S)_{ij} =
\begin{cases}
1 & \text{if } (i,j)\in S,\\
0 & \text{otherwise},
\end{cases}
\]
and let $\bar S = [1,m]^2 \setminus S$ be its complement, with associated matrix $M_{\bar S}$. Then the Hilbert series of a PBW algebra $A$ takes the form
\[
h_A(z)=\frac{\det(1+zM_{\bar S})}{\det(1-zM_S)}.
\]

We conclude that for every generating function $G(t)$ arising in Theorem~2, there exists an integer $m$ and a subset $S \subseteq [1,m]^2$ such that
\[
G(t)=\frac{\det(1+tM_{\bar S})}{\det(1-tM_S)}.
\]

Thus, the admissible class of Hilbert–Poincaré series is exactly realised by ratios of characteristic determinants of $0$--$1$ matrices associated with complementary subsets of $[1,m]^2$. In particular, the analytic structure of $G(t)$ is encoded by the spectral data of the pair $(M_S, M_{\bar S})$, providing a concrete combinatorial realisation of the abstract classification given by Theorem~2.

It is worth placing this result as well in the context of classical results on total positivity. A sequence $(a_n)_{n\ge 0}$ is a Pólya frequency (PF) sequence if and only if its associated Toeplitz matrix is totally nonnegative, that is, if and only if it is totally positive. In \cite{ReinerWelker2005}, the authors ask for which Koszul algebras the Hilbert function is a PF sequence. Within our framework, Theorem~2 provides an explicit construction of a large class of such sequences: the single-mode Hilbert--Poincaré series arising from PBW $U(d)$-equivariant algebras are precisely of the form $G(t)=Q_-(t)/Q_+(t)$ with $Q_\pm$ having real roots of fixed sign, and therefore generate totally positive (PF) coefficient sequences of the rational type. In this sense, our result realises a distinguished subclass of PF Hilbert functions through a concrete algebraic and physically meaningful construction keeping in mind that these series by \cite{MedinaSanchezDakic2024} are in one-to-one correspondence with partition functions. In that sense we have stablished a complete formal classification of the link between partition functions, state spaces and as we will see in the next section, creation-annihilation algebras.
\section{Creation--annihilation algebras and transfields}\label{sec:creation_annihilation}

The preceding analysis has determined the algebraic skeleton of the many-particle space: a quadratic quotient of the tensor algebra, constrained by unitarity and the PBW condition, whose single-mode Hilbert-Poincare series characterises the statistics completely. What this construction provides, however, is only the space of states together with the action of the creation operators. A full quantum mechanical framework requires more: one needs annihilation operators, their commutation relations with the creators, and a representation of the observable algebra on the resulting Fock space.\\
The passage from the creation algebra to the full creation-annihilation algebra is not automatic. It requires specifying how an annihilation operator moves past a creation operator, and this exchange rule must be compatible with the quadratic relations already imposed on each sector separately. The data that controls this compatibility is a linear map between the mixed tensor products of the one-particle space and its dual, which we call the cross map. Its admissible forms are constrained by Yang-Baxter type identities inherited from the PBW condition, and its non-homogeneous part is fixed by the vacuum two-point function. The present section develops this construction in full and concludes by showing that the resulting algebra carries a natural representation of the general linear group, providing the transfield analogue of the standard Fock representation.

\subsection{Annihilation algebra}
Let $\{X_{i\alpha}\}$ be the chosen homogeneous basis of $H = \mathbb{C}^{d}\otimes K$,
and let $\{X^{\ast}_{i\alpha}\}$ be the dual basis of $H^{\ast}$, defined by
\[
X^{\ast}_{i\alpha}(X_{j\beta}) \;=\; \delta_{ij}\, g_{\alpha\beta},
\]
where $g_{\alpha\beta}$ is a nondegenerate Hermitian form on $K$.  This pairing
extends to a nondegenerate bilinear form between $H^{\ast}\otimes H^{\ast}$ and
$H\otimes H$.

Throughout Section IV we fix the quadratic data determined in Section III.
Namely, let
\[
R := \operatorname{im} P_g \subset H \otimes H,
\qquad
R^\perp := \ker P_g,
\]
so that the creation algebra is
\[
A := T(H)/\langle R \rangle.
\]
We define the dual quadratic subspace
\[
S := R^* = \operatorname{im} P_g^* \subset H^* \otimes H^*,
\]
and the annihilation algebra
\[
B := T(H^*)/\langle S \rangle.
\]
Finally, we fix a linear map
\[
C : H^* \otimes H \to H \otimes H^*
\]
(the \emph{cross map}) which will encode the exchange relations between
annihilation and creation operators and is required to satisfy the
compatibility conditions of Theorem 3.

The quadratic data for creation operators consist of the decomposition
\[
H\otimes H \;=\; R_{\mathrm{g}} \;\oplus\; R_{\mathrm{g}}^{\perp},
\qquad
P_{\mathrm{g}} + P_{\mathrm{g}}^{\perp} = \mathrm{id},
\]
where $R_{\mathrm{g}} = \mathrm{im}\,P_{\mathrm{g}}$ is the space of quadratic
relations and $R_{\mathrm{g}}^{\perp}$ is the space of admissible ordered
quadratic monomials.  Under the dual pairing, this orthogonal decomposition
induces a corresponding decomposition
\[
H^{\ast}\otimes H^{\ast}
\;=\;
R_{\mathrm{g}}^{\ast} \;\oplus\; (R_{\mathrm{g}}^{\perp})^{\ast},
\]
where $R_{\mathrm{g}}^{\ast}$ is the annihilator of $R_{\mathrm{g}}^{\perp}$ and is
precisely the image of the adjoint projector
\[
P_{\mathrm{g}}^{\ast} \colon H^{\ast}\otimes H^{\ast} \longrightarrow
H^{\ast}\otimes H^{\ast},
\]
defined by
\[
\langle P_{\mathrm{g}}^{\ast}(\eta),\, v\rangle 
\;=\; 
\langle \eta,\, P_{\mathrm{g}}(v)\rangle,
\qquad
\eta\in H^{\ast}\otimes H^{\ast},\; v\in H\otimes H.
\]

The annihilation relations are therefore encoded by the dual quadratic subspace
\[
R_{\mathrm{g}}^{\ast} = \mathrm{im}\,P_{\mathrm{g}}^{\ast},
\]
and the annihilation algebra is the quadratic quotient
\[
B \;:=\; T(H^{\ast})\big/ \langle R_{\mathrm{g}}^{\ast}\rangle.
\]

This construction shows that the creation--annihilation algebras are related
by orthogonality: the projector $P_{\mathrm{g}}$ selects which quadratic
combinations of creation operators vanish, while its orthogonal complement
$P_{\mathrm{g}}^{\perp}$ selects the ordered monomials that survive.  Under the
dual pairing these roles are reversed, so the annihilation relations mirror the
creation relations with arrows reversed.  In particular, the combinatorial
structure governing annihilation is determined by the orthogonality between
$R_{\mathrm{g}}$ and $R_{\mathrm{g}}^{\perp}$.

\subsection{The Creation--Annihilation Algebras}

To obtain a single associative algebra containing both creation and annihilation operators we follow a construction due to Borowiec and Marcinek. The basic ingredient is a linear map
\[
C\colon H^{\ast}\otimes H \longrightarrow H\otimes H^{\ast},
\]
which specifies the exchange rule between a single annihilation generator and a single creation generator. Extending \(C\) multiplicatively to the free algebras \(T(H^{\ast})\) and \(T(H)\) yields a map
\[
\tilde{\tau}\colon T(H^{\ast})\otimes T(H) \to T(H)\otimes T(H^{\ast}),
\]
which is required to descend to a map
\[
\tau\colon B\otimes A \to A\otimes B
\]
on the quadratic quotients. The necessary and sufficient conditions are that the quadratic ideals be stable under the cross relations.

\begin{theorem}[Borowiec--Marcinek, \cite{BM2000}]
Let \(A = T(H)/I_{R}\) and \(B = T(H^{\ast})/I_{S}\) be quadratic algebras defined by ideals generated by \(R\subset H\otimes H\) and \(S\subset H^{\ast}\otimes H^{\ast}\). A homogeneous cross map \(\tilde{\tau}\) as above descends to a well defined map \(\tau\colon B\otimes A\to A\otimes B\) if and only if
\begin{enumerate}
  \item \(I_{R}\) is a left \(\tilde{\tau}\)--ideal, that is, \(\tilde{\tau}(H^{\ast}\otimes I_{R})\subset I_{R}\otimes T(H^{\ast})\),
  \item \(I_{S}\) is a right \(\tilde{\tau}\)--ideal, that is, \(\tilde{\tau}(I_{S}\otimes H)\subset T(H)\otimes I_{S}\).
\end{enumerate}
Equivalently, these conditions can be expressed in terms of Yang--Baxter type identities involving the quadratic data and the map \(C\).
\end{theorem}
Here \(R \subset H\otimes H\) and \(S \subset H^*\otimes H^*\) denote the
quadratic relation subspaces defining the creation and annihilation algebras,
while \(C : H^*\otimes H \to H\otimes H^*\) is the cross map introduced above.

Under these hypotheses one can form an algebra structurally isomorphic, as a vector space, to \(A\otimes B\), in which the mixed relations are encoded by \(C\). This algebra contains both the pure creation and pure annihilation algebras and is associative by construction.

\subsection{Vacuum two--point function and non--homogeneous term}
Up to this point the construction determines only the homogeneous, quadratic part of the mixed relations. But this does not yet take into account the behaviour of the vacuum.

In a Fock representation the vacuum is characterised by the two--point function
\[
\langle 0 | X_{i\alpha} X^{\dagger}_{j\beta} |0\rangle.
\]
We assume that different external modes are orthogonal and that the internal degrees of freedom are paired by the Hermitian form \(g\) as we introduced in the previous section. Thus we set
\[
\langle 0 | X_{i\alpha} X^{\dagger}_{j\beta} |0\rangle \;=\; g_{\alpha\beta}\,\delta_{ij}.
\]
This defines a linear functional
\[
G\colon H^{\ast}\otimes H\to \mathbb{C}, \qquad
G(X^{\ast}_{i\alpha}\otimes X_{j\beta}) = g_{\alpha\beta}\,\delta_{ij}.
\]
Borowiec and Marcinek observed that, because \(G\) takes values in the centre of the algebra, adding \(G\) to the homogeneous cross
\[
\tau = \tilde{\tau} + G
\]
does not affect the quadratic consistency conditions. The resulting mixed relations can be written compactly using a suitable bracket that depends on the internal projectors.

\subsection{The \texorpdfstring{\((A,B)\)}{(A,B)} bracket}

Let \(A\) and \(B\) denote operators determined by the internal projectors \(P^{(K)}_{\mathrm{sym}}\) and \(P^{(K)}_{\mathrm{ext}}\). For any pair of operators \(P_{i\alpha}\) and \(Q_{j\beta}\) we define
\begin{align}
\label{eq:ABbracket}
[ P_{i\alpha}, Q_{j\beta} ]_{A,B}
:= 
&\sum_{\eta,\lambda}
A^{\eta\lambda}{}_{\alpha\beta}
\bigl[ P_{i\eta}, Q_{j\lambda} \bigr]_{+}
\;+\\
&\sum_{\eta,\lambda}
B^{\eta\lambda}{}_{\alpha\beta}
\bigl[ P_{i\eta}, Q_{j\lambda} \bigr]_{-},
\end{align}
where $[P_{i\alpha}, Q_{j\beta}]_{\pm}$ represent the commutator ($-$)
and anticommutator ($+$). When \(P\) and \(Q\) are both creation operators, this bracket encodes the quadratic relations among them; similarly for annihilation operators. The mixed relations can then be written as
\begin{align}
[X_{i\alpha}^\dagger, X^{\dagger}_{j\beta}]_{A,B} &= 0\\
[X_{i\alpha}, X_{j\beta}]_{A,B} &= 0\\
[X_{i\alpha}, X^{\dagger}_{j\beta}]_{A,B} &= g_{\alpha\beta}\,\delta_{ij} \,\mathbf{1},
\end{align}
which generalise the standard CCR and CAR. The values of \(A\) and \(B\) are fixed by the Yang--Baxter compatibility conditions relating the creation, annihilation, and cross exchanges.

A useful special case is obtained when the tensors \(A\) and \(B\) arise from a single operator
\[
R:K\otimes K\to K\otimes K
\]
by
\[
A=\frac{1-R}{2},\qquad B=\frac{1+R}{2}.
\]
Then \(A+B=\mathrm{id}\) and \(A-B=-R\), so the \((A,B)\)-bracket takes the form
\[
[P_{i\alpha},Q_{j\beta}]_{A,B}
=
P_{i\alpha}Q_{j\beta}
-
\sum_{\eta,\lambda}R^{\eta\lambda}_{\alpha\beta}\,Q_{j\lambda}P_{i\eta}.
\]
Accordingly, the defining relations become
\[
X^\dagger_{i\alpha}X^\dagger_{j\beta}
=
\sum_{\eta,\lambda}R^{\eta\lambda}_{\alpha\beta}\,
X^\dagger_{j\lambda}X^\dagger_{i\eta},
\qquad
X_{i\alpha}X_{j\beta}
=
\sum_{\eta,\lambda}R^{\eta\lambda}_{\alpha\beta}\,
X_{j\lambda}X_{i\eta},
\]
and
\[
X_{i\alpha}X^\dagger_{j\beta}
=
\sum_{\eta,\lambda}R^{\eta\lambda}_{\alpha\beta}\,
X^\dagger_{j\lambda}X_{i\eta}
+
g_{\alpha\beta}\delta_{ij}\mathbf 1.
\]
After a relabelling of internal indices, this is precisely the form of the R-paraparticle exchange relations \cite{WangHazzardProvided}. That framework is therefore recovered as a special case of the present bracket formalism, namely when the quadratic data are encoded by a single involutive Yang--Baxter operator \(R\).

\subsection{Transfield representation of \texorpdfstring{\(\mathfrak{gl}(d)\)}{gl(d)}}

Using the above relations we can now construct a representation of the general
linear algebra in terms of the creation and annihilation operators defining our
transtatistical fields (namely \emph{Transfields}). We introduce a family of operators \(J_{ij}\) defined using the internal Hermitian
form \(g\). Let \(g_{\alpha\beta}\) be a nondegenerate Hermitian form on the internal
space \(K\), and let \(g^{\beta\alpha}\) denote the entries of its inverse, so that
\(\sum_{\beta} g_{\alpha\beta} g^{\beta\gamma} = \delta^\gamma_{\alpha}\). We define
\[
J_{ij} := \sum_{\alpha,\beta} g^{\beta\alpha}\, X^{\dagger}_{i\alpha} X_{j\beta},
\]
where \(i,j = 1,\dots,d\) label the external modes and \(\alpha,\beta\) the internal
degrees of freedom. We say that the Hermitian form $g$ on $K$ is compatible with the quadratic
data $(R,S,C)$ if the projector $P_g$ is self-adjoint with respect to the
induced inner product on $H \otimes H$, and the cross map $C$ is adjoint to
the quadratic data under the pairing induced by $g$.
Under these adjointness conditions, we may state the following lemma.

\begin{lemma}
\label{lem:gl_d_action}
The following identities hold:
\begin{enumerate}
  \item[(i)] (Action on creators)
  \[
    [J_{ij}, X_{k\sigma}^\dagger]
    = \delta_{jk}\,X_{i\sigma}^\dagger,
    \qquad
    1 \leq i,j,k \leq d.
  \]
  \item[(ii)] (Action on annihilators)
  \[
    [J_{ij}, X_{k\sigma}]
    = -\,\delta_{ik}\,X_{j\sigma},
    \qquad
    1 \leq i,j,k \leq d.
  \]
  \item[(iii)] (Commutators among the $J$)
  \[
    [J_{ij}, J_{kl}]
    = \delta_{jk}\,J_{il}
      - \delta_{il}\,J_{kj},
    \qquad
    1 \leq i,j,k,l \leq d.
  \]
\end{enumerate}
In particular, the operators $J_{ij}$ satisfy the commutation relations of the
Lie algebra $\mathfrak{gl}(d)$.
\end{lemma}

This assumption is well motivated by the fact that the operators \(R\), \(S\), and \(C\)
encode, in a precise algebraic sense, the exchange of indistinguishable particles.
Accordingly, their action on multiparticle states should not alter transition
probabilities, that is, the two–particle correlations and, consequently, the
correlations in systems with an arbitrary number of particles.

We thus see that the operators $J_{ij}$ satisfy the commutation relations of the
general linear algebra $\mathfrak{gl}(d)$, providing a natural way to construct
representations of $\mathfrak{gl}(d)$ in terms of the transfields
$X_{i\alpha},\,X_{i\alpha}^\dagger$.

\section{An explicit finite--order example}\label{sec:example}

To illustrate the abstract formalism we briefly describe a simple example with a finite maximal occupation number for each mode.

We keep the internal space $K \cong \mathbb C^3$ with orthonormal basis
$\{k_\alpha\}_{\alpha=1}^3$ and generators
\[
X_{i\alpha} := e_i \otimes k_\alpha
\qquad (i=1,2)
\]
where $i$ labels the external modes.  
As before, we choose the unique surviving quadratic
direction in $K\otimes K$ to be
\[
h := \frac{1}{\sqrt{3}}\bigl(
k_1\otimes k_1 + k_2\otimes k_2 + k_3\otimes k_3
\bigr),
\]
and let $R_{\mathrm{g}} := h^\perp \subset K\otimes K$ generate the
quadratic relations for a single mode.  

For each fixed $i$, this implies
that all off--diagonal products $X_{i\alpha}X_{i\beta}$ with $\alpha\neq\beta$
vanish and all diagonal products $X_{i\alpha}X_{i\alpha}$ are identified in
the quotient.  
Thus the single--mode Hilbert--Poincar\'e series is
\[
G(t) = 1 + 3 t + t^2,
\]
so that for $d=2$ modes we obtain
\[
H_F(t) = G(t)^2 = 1 + 6 t + 11 t^2 + 6 t^3 + t^4.
\]
In particular, the global Fock space $F$ is supported in total degrees
$0,\dots,4$, and we have the factorization
\[
F \;\cong\; W^{\otimes 2},
\]
where $W$ is the single--mode Fock space with homogeneous components
$\dim W_0 = 1$, $\dim W_1 = 3$, $\dim W_2 = 1$, and $W_n = 0$ for $n\ge 3$.

Let $\ket{0}_i$ denote the vacuum of mode $i$ and set
$\ket{0} := \ket{0}_1\otimes \ket{0}_2$.
For each mode $i$ we choose the following homogeneous basis of $W$:
\[
W_0 = \mathrm{span}\{\ket{0}_i\},\qquad
W_1 = \mathrm{span}\{\hat X_{i\alpha}^\dagger\ket{0}_i \mid \alpha=1,2,3\},
\]
\[
W_2 = \mathrm{span}\{Y_i\}, \qquad
Y_i := \bigl(X_{i1}X_{i1} + X_{i2}X_{i2} + X_{i3}X_{i3}\bigr),
\]
so that the vector corresponding to $Y_i$ in the Fock representation is
\[
\hat Y_i^\dagger\ket{0}_i
=
\bigl(
\hat X_{i1}^\dagger \hat X_{i1}^\dagger
+
\hat X_{i2}^\dagger \hat X_{i2}^\dagger
+
\hat X_{i3}^\dagger \hat X_{i3}^\dagger
\bigr)\ket{0}_i.
\]

Using the tensor product decomposition $F \cong W_1 \otimes W_2$, the
homogeneous subspaces of the full Fock space and their bases up to total
degree $4$ are as follows.

\medskip\noindent
\begin{itemize}

\item \textbf{Degree $0$.}
\[
F_0 = \mathrm{span}\{\ket{0}\},\qquad
\ket{0} = \ket{0}_1\otimes\ket{0}_2.
\]

\medskip\noindent
\item \textbf{Degree $1$.}
Here $N = p+q$ with $(p,q)\in\{(1,0),(0,1)\}$, hence
\[
k_1 \cong (W_1\otimes W_0)\oplus(W_0\otimes W_1)
\]
and a basis is
\[
\hat X_{1\alpha}^\dagger\ket{0} := 
\hat X_{1\alpha}^\dagger\ket{0}_1\otimes\ket{0}_2,\qquad
\hat X_{2\alpha}^\dagger\ket{0} := 
\ket{0}_1\otimes \hat X_{2\alpha}^\dagger\ket{0}_2,
\]
for $\alpha=1,2,3$.  Thus $\dim k_1 = 3+3 = 6$.

\medskip\noindent
\item \textbf{Degree $2$.}
Now $N=2$ with $(p,q)\in\{(2,0),(1,1),(0,2)\}$, so
\[
k_2 \cong (W_2\otimes W_0)\;\oplus\;(W_1\otimes W_1)\;\oplus\;(W_0\otimes W_2).
\]
A convenient basis is given by:
\begin{itemize}
\item one vector of type $(2,0)$:
\[
\hat Y_1^\dagger\ket{0} := \hat Y_1^\dagger\ket{0}_1\otimes\ket{0}_2,
\]
\item nine vectors of type $(1,1)$:
\[
\hat X_{1\alpha}^\dagger \hat X_{2\beta}^\dagger\ket{0}
:= 
\hat X_{1\alpha}^\dagger\ket{0}_1\otimes 
\hat X_{2\beta}^\dagger\ket{0}_2,
\qquad \alpha,\beta=1,2,3,
\]
\item one vector of type $(0,2)$:
\[
\hat Y_2^\dagger\ket{0} := 
\ket{0}_1\otimes \hat Y_2^\dagger\ket{0}_2.
\]
\end{itemize}
Hence $\dim k_2 = 1 + 9 + 1 = 11$, in agreement with the coefficient of
$t^2$ in $H_F(t)$.

\medskip\noindent
\item \textbf{Degree $3$.}
Here $N=3$ with $(p,q)\in\{(2,1),(1,2)\}$, so
\[
k_3 \cong (W_2\otimes W_1)\;\oplus\;(W_1\otimes W_2).
\]
A basis can be chosen as
\begin{align*}
\hat Y_1^\dagger \hat X_{2\beta}^\dagger\ket{0}
&:= 
\hat Y_1^\dagger\ket{0}_1\otimes 
\hat X_{2\beta}^\dagger\ket{0}_2,
&&\beta=1,2,3,
\\[2pt]
\hat X_{1\alpha}^\dagger \hat Y_2^\dagger\ket{0}
&:= 
\hat X_{1\alpha}^\dagger\ket{0}_1\otimes 
\hat Y_2^\dagger\ket{0}_2,
&&\alpha=1,2,3.
\end{align*}
Thus $\dim k_3 = 3+3 = 6$.

\medskip\noindent
\item \textbf{Degree $4$.}
Finally $N=4$ forces $(p,q)=(2,2)$, so
\[
F_4 \cong W_2\otimes W_2
\]
is one--dimensional, with basis
\[
\hat Y_1^\dagger \hat Y_2^\dagger \ket{0}
:= 
\hat Y_1^\dagger\ket{0}_1 \otimes \hat Y_2^\dagger\ket{0}_2.
\]
\end{itemize}
With this, we can finally write the relations that generate the algebra.  
For each mode $i \in \{1,2\}$ the quadratic ideal generated by 
$R_{\mathrm{g}} = h^\perp \subset K \otimes K$ is given by
\begin{equation}\label{eq:single-mode-rel}
X_{i\alpha}X_{i\beta} = 0
\quad (\alpha \neq \beta),
\qquad
X_{i1}X_{i1} = X_{i2}X_{i2} = X_{i3}X_{i3}.
\end{equation}
Equivalently, a basis of quadratic relations for mode $i$ is
\begin{align}
X_{i\alpha}X_{i\beta} &= 0 
\quad (\alpha \neq \beta), 
\\
X_{i1}X_{i1} - X_{i2}X_{i2} &= 0,
\\
X_{i1}X_{i1} - X_{i3}X_{i3} &= 0.
\end{align}

This example exhibits a non--trivial finite--order statistics that differs both from bosons and fermions yet fits within the general quadratic quotient framework. It provides a concrete realisation of a transtatistical species.

\section{Conclusions and outlook}\label{sec:outlook}

We have developed an operationally motivated generalisation of the transition from first to second quantization, without assuming symmetrisation principles as fundamental postulates. Building on the work previously presented by the authors on the reconstruction of quantum statistics, we have provided an explicit algebraic formulation of families of transtatistics using quotients of tensor algebras over single-particle spaces. The assumption that an ordered basis of monomials survives the quotient then yields the operational locality condition, expressed as mode-by-mode independence of particle counting. Once this structural requirement is in place, we impose a U(d) symmetry on the accessible degrees of freedom, which fixes the external behaviour of the relations and supplies the representation-theoretic framework necessary for defining transtatistics.

From a physical perspective, this quotient-based construction offers a way of defining indistinguishability in operational terms, as the loss of accessible information that could be used to differentiate between particles. Algebraically, the emergence of the Yang--Baxter equation as a necessary condition for statistical consistency links directly to the operational notion of locality: once states are defined solely through particle numbers, relative to additional information that is operationally inaccessible, the resulting space is quadratically generated, with the Yang--Baxter solutions ensuring invariance of the statistics under subsystem decompositions.

On the other hand, although we have defined the quadratic relations that determine the Fock space through the creation operators, the inner product on $\mathcal{H}$ allows us to introduce the annihilation operators in the standard way, namely as the adjoints of the creation operators. In the same way, we use this notion of orthogonality to formalize the creation--annihilation algebras. Moreover, this framework connects with existing work on parastatistics, which share the guiding idea of scrambling internal degrees of freedom that are operationally inaccessible, where we find the classical examples \cite{Green53} together with more recent developments \cite{WangHazzardProvided}. In those cases the relations are introduced heuristically and by postulate, whereas here we obtain a complete characterization derived from basic operational assumptions about the quantum physics of many-particle systems.

From a broader mathematical perspective, the present construction contributes to the problem of identifying Koszul algebras whose Hilbert functions are totally positive (Pólya frequency) sequences, as formulated in \cite{ReinerWelker2005}. Since PF sequences coincide with totally positive sequences, this question can be viewed as a structural problem at the interface of commutative algebra, combinatorics, and representation theory. Recent work \cite{SamVandeBogert2024} further highlights deep connections between total positivity and algebraic constructions such as Schur-type modules and free resolutions.

The class of PBW $U(d)$-equivariant algebras introduced here provides a concrete and physically motivated family for which total positivity holds and is completely controlled by quadratic data. This suggests that the operational principles underlying indistinguishability naturally single out a rigid subclass of Koszul-type algebras in which total positivity is not merely a property, but a consequence of the underlying algebraic and symmetry structure.

A natural generalisation of this work is its application to quantum field theory. Since we have introduced new commutation relations, one can move from discrete degrees of freedom to measure spaces such as those modelling spacetime. This would allow the description of new types of fields compatible with transtatistics, while remaining consistent with the operational arguments presented here regarding accessible information and operational indistinguishability, but now in the continuum. Of particular interest are the consequences of imposing a causal structure on the accessible degrees of freedom, as this additional symmetry could further constrain, or even fully collapse, the Yang--Baxter symmetry of the projectors to the standard bosonic and fermionic cases. This can be viewed as an extension of the spin--statistics theorem within the present framework \cite{Pauli1940,StreaterWightman1964,Greenberg1991}.
\section*{Acknowledgements}
We thank Tristan Malleville and Julian Maisriml for their comments and for insightful discussions on the subject matter. This research was funded in whole, or in part, by the Austrian Science Fund (FWF) [10.55776/F71], [10.55776/P36994] and [10.55776/COE1] and the European Union – NextGenerationEU. For open access purposes, the author(s) has applied a CC BY public copyright license to any author accepted manuscript version arising from this submission.
\section*{Conflict of Interest}
The author declares that there is no conflict of interest.

\section*{Data Availability}
No data were generated or analyzed in this study.
\bibliography{apssamp}

\onecolumngrid
\appendix

\appendix

\section{Algebraic preliminaries}

\subsection{Ordered basis implies quadratic generation}
 \begin{proof}[Lemma 1]
  Given a homogeneous ideal \(I \subset T(H)\) and a fixed monomial order on \(T(H)\), assume that there exists an ordered monomial basis for the quotient. Without loss of generality, we may then take the set
\[
\mathcal{B} := \{\, x_1^{n_1}\cdots x_d^{n_d} \mid n_i \in \mathbb{N} \,\},
\]
where \(\{x_1,\ldots,x_d\}\) is a basis of \(H\), and regard its images
\[
[x_1^{n_1}\cdots x_d^{n_d}] \in T(H)/I
\]
as spanning the quotient. Any unordered monomial \(w \in T(H)\) is projected to a linear combination of ordered monomials; equivalently, the projection induces a rewriting map
\[
w \longmapsto \sum_{u \in \mathcal{B}} c_u\, u,
\]
where the coefficients \(c_u\) are determined by the defining relations of \(I\). In particular, the degree–two monomials carry a prescribed ordering rule, and the projection corrects any violation of this rule.

Consider now an arbitrary monomial \(w\) of degree \(N > 2\). One may attempt to rewrite \(w\) into ordered form by iteratively correcting the order of all degree–two subwords. Under the assumption that the rewriting system generated by the degree–two relations is confluent, this pairwise correction uniquely determines an ordered representative of \(w\).

Assume, for contradiction, that in addition to the degree–two relations there exists an independent higher–degree relation
\[
r = \sum_i \alpha_i w_i \in I,
\]
with all \(w_i\) of degree \(N > 2\), which is not a consequence of the degree–two relations. Then a monomial of degree \(N\) appearing in \(r\) admits two distinct reductions:
\[
w \overset{\text{pairwise}}{\longmapsto} u_{\mathrm{pair}}  
\qquad\text{and}\qquad
w \overset{r}{\longmapsto} u_{\mathrm{high}},
\]
where \(u_{\mathrm{pair}}\) is the ordered monomial obtained by successive degree–two corrections, and \(u_{\mathrm{high}}\) is the expression obtained by applying the higher–degree relation. If \(u_{\mathrm{pair}} \neq u_{\mathrm{high}}\), then the difference provides a nontrivial linear dependence among ordered monomials:
\[
u_{\mathrm{pair}} - u_{\mathrm{high}} = 0 \quad \text{in } T(H)/I,
\]
contradicting the assumption that the ordered monomials form a basis.

Hence no additional higher–degree relations may exist independently of those generated by degree–two relations.
 \end{proof}

\subsection{Factorisation as a tensor power}
 \begin{proof}[Lemma 2]
     Let \(\{X_{i\alpha}\}\) be a basis of \(H\). Assume a monomial order in which all monomials are grouped into blocks according to the index \(i\). That is, the quotient \(T(H)/I\) is spanned by elements of the form
\[
[X_{1\alpha_{1,1}}\cdots X_{1\alpha_{1,m_1}}
 \, X_{2\alpha_{2,1}}\cdots X_{2\alpha_{2,m_2}}
 \cdots
 X_{d\alpha_{d,1}}\cdots X_{d\alpha_{d,m_d}}].
\]
Equivalently, we may describe any ordered monomial by writing it as
\[
w_1 w_2 \cdots w_d,
\]
where each word \(w_i\) is generated by the symbols \(X_{i\alpha}\), with \(\alpha\) varying and \(i\) fixed.

Since we have assumed that the ordered monomials form a basis of the quotient, the set of all ordered monomials with fixed \(i\) spans a subspace of the quotient. Denote this subspace by
\[
A_i \subset T(H)/I.
\]

There is a natural map
\[
A_1 \otimes \cdots \otimes A_d \longrightarrow T(H)/I,
\qquad
w_1 \otimes \cdots \otimes w_d \longmapsto [w_1\cdots w_d],
\]
defined by concatenation followed by projection onto the quotient.

This map is surjective, since every ordered monomial in the quotient has the form \(w_1\cdots w_d\) as described. It is also injective: the tensors \(w_1\otimes\cdots\otimes w_d\) form a basis of \(A_1\otimes\cdots\otimes A_d\), and the map sends this basis to a basis of the quotient. Hence we obtain a vector space isomorphism
\[
T(H)/I \cong A_1\otimes\cdots\otimes A_d.
\]

Moreover, by assumption the degeneracies for each mode \(i\) are identical, so all subspaces \(A_i\) are mutually isomorphic. Let \(W\) denote their common isomorphism class. Then
\[
T(H)/I \cong W^{\otimes d}.
\]

Finally, there is a natural \(U(1)\)-action on the quotient: an ordered monomial is multiplied by a phase whose argument depends on its degree. For $g\in U(1)$ we have $g\cdot w\coloneqq e^{in\theta}w$, for all $w$ of order $n$. In the same way, we have an action of $U(1)$ on $W^{\otimes d}$ similarly by adding up the degrees of the components on $w_1\otimes...\otimes w_d$ and then multiplying by a phase accordingly. Naturally, these actions commute with respect to the isomorphism hence we have that
\[
T(H)/I \cong W^{\otimes d}
\]
as \(U(1)\)-representations.

 \end{proof}

\subsection{$U(d)$-invariant quadratic subspaces}
 \begin{proof}[Lemma 3]
We begin with the decomposition
\[
H \otimes H 
\;\cong\;
\mathrm{Sym}^2(\mathbb{C}^d)\otimes (K\otimes K)
\;\oplus\;
\wedge^2(\mathbb{C}^d)\otimes (K\otimes K).
\]
Each summand is an isotypic component of the \(U(d)\)-representation on \(H\otimes H\).

By Schur's lemma, any \(U(d)\)-equivariant operator
\[
T : H\otimes H \longrightarrow H\otimes H
\]
must act as the identity on the irreducible factor tensored with an arbitrary linear map on the multiplicity space \(K\otimes K\). In particular, if \(R \subset H\otimes H\) is a \(U(d)\)-invariant subspace, then there exists a \(U(d)\)-equivariant projection
\[
P : H\otimes H \longrightarrow H\otimes H
\]
whose image is \(R\).

With respect to the isotypic decomposition above, the projector decomposes as
\[
P = 
\bigl(\mathrm{id}_{\mathrm{Sym}^2(\mathbb{C}^d)} \otimes S_{\mathrm{sym}}\bigr)
\;\oplus\;
\bigl(\mathrm{id}_{\wedge^2(\mathbb{C}^d)} \otimes S_{\mathrm{ext}}\bigr),
\]
where
\[
S_{\mathrm{sym}},\; S_{\mathrm{ext}} : K\otimes K \longrightarrow K\otimes K
\]
are linear projections.

Define
\[
W_{\mathrm{sym}} = \operatorname{im}(S_{\mathrm{sym}}),
\qquad
W_{\mathrm{ext}} = \operatorname{im}(S_{\mathrm{ext}}).
\]
Then the image of \(P\), which is precisely the invariant subspace \(R\), has the form
\[
R
=
\mathrm{Sym}^2(\mathbb{C}^d)\otimes W_{\mathrm{sym}}
\;\oplus\;
\wedge^2(\mathbb{C}^d)\otimes W_{\mathrm{ext}},
\]
for certain subspaces \(W_{\mathrm{sym}}, W_{\mathrm{ext}} \subset K\otimes K\).

Thus every \(U(d)\)-invariant subspace of \(H\otimes H\) is obtained by choosing arbitrary multiplicity subspaces in each isotypic component.

 \end{proof}

\section{PBW property and Yang--Baxter constraints}

This is based on \cite{PolishchukPositselski2005}.
\subsection*{PBW Bases}

Let \(A = T(V)/I\) be a quadratic algebra, where \(V\) has an ordered basis
\(\{x_1,\dots,x_m\}\). For each multiindex \(\alpha=(i_1,\dots,i_n)\), write
\[
x^\alpha = x_{i_1}\cdots x_{i_n}\in T(V),
\]
and impose a lexicographic order on all multiindices.

The construction begins by identifying the degree--two monomials that survive modulo the quadratic relations. Given the order on multiindices, one considers the set
\[
S^{(2)} = \{(i,j)\mid x_i x_j \text{ cannot be written modulo } I 
\text{ as a linear combination of } x_k x_\ell \text{ with } (k,\ell)<(i,j)\}.
\]
These are the ``admissible'' pairs.

Higher--degree admissible multiindices are then defined inductively:
\[
S^{(n)} = \{(i_1,\dots,i_n) \mid (i_k,i_{k+1})\in S^{(2)} 
\text{ for all } k\}.
\]
Thus a monomial \(x_{i_1}\cdots x_{i_n}\) is admissible precisely when each of its adjacent pairs is admissible in degree two.

\begin{definition}[PBW Basis]
The algebra \(A=T(V)/I\) is said to admit a \emph{PBW basis} if the set
\[
\mathcal{B}
= \{x^\alpha \mid \alpha \in \bigcup_{n\ge 0} S^{(n)}\}
\]
projects to a vector space basis of \(A\).  
In this case, the generators \(x_1,\dots,x_m\) are called \emph{PBW generators} of \(A\).
\end{definition}

Thus a PBW basis is a monomial basis, compatible with lexicographic order, such that all nonadmissible monomials reduce uniquely to linear combinations of admissible ones. It is the noncommutative analogue of the classical Poincar\'e--Birkhoff--Witt basis.

The key structural fact is that if a pair \((i,j)\notin S^{(2)}\), then
\[
x_i x_j
= \sum_{(k,\ell)<(i,j)} c^{k,\ell}_{i,j} \, x_k x_\ell,
\]
and these quadratic reductions extend consistently to all degrees through iterated projections. A PBW basis exists precisely when these reductions are confluent.

\subsection*{PBW Theorem}

\begin{theorem}[PBW Theorem]
If the cubic monomials \(x_{i_1} x_{i_2} x_{i_3}\) with 
\((i_1,i_2,i_3)\in S^{(3)}\) are linearly independent in degree three of
\(A=T(V)/I\), then the same holds in every degree.  
Therefore, the generators \((x_1,\dots,x_m)\) are PBW generators and 
\(\mathcal{B}\) is a PBW basis of \(A\).
\end{theorem}

Equivalently, the cubic independence assumption is encoded by the braid-type compatibility condition
\[
\pi^{12}\pi^{23}\pi^{12} = \pi^{23}\pi^{12}\pi^{23},
\]
where the maps \(\pi^{ij}\) apply the quadratic reduction to the \(i\)-th and \(j\)-th tensor positions of \(V^{\otimes 3}\). This operator identity ensures that all overlaps of quadratic relations resolve consistently, guaranteeing global confluence of the rewriting system.

 \begin{proof}
     The result follows by direct application of the PBW theorem.
 \end{proof}

\section{Single-mode structure and proof of the main theorem}

\subsection{Construction of the spaces $W_n$}
Let $K$ be the internal space and let
\[
P_{\mathrm{g}} : K \otimes K \longrightarrow K \otimes K
\]
be the quadratic projector determining the relations. For each $n \geq 0$ define
\[
V_n := K^{\otimes n}.
\]

For $1 \leq k \leq n-1$ we introduce the operators
\[
P^{(n)}_k := \mathrm{id}^{\otimes (k-1)} \otimes P_{\mathrm{g}} \otimes \mathrm{id}^{\otimes (n-k-1)} : V_n \longrightarrow V_n,
\]
acting nontrivially only on the $k$-th and $(k+1)$-th tensor factors.

The space of admissible degree-$n$ monomials is then defined as
\[
W_n := \bigcap_{k=1}^{n-1} \ker P^{(n)}_k \subseteq K^{\otimes n}.
\]

Equivalently, $W_n$ consists of those tensors in $K^{\otimes n}$ that survive all quadratic reductions, i.e. those configurations for which no adjacent pair lies in the generating subspace $\mathrm{im}\,P_{\mathrm{g}}$.

By construction, the family $\{W_n\}_{n\geq 0}$ defines a graded vector space
\[
W := \bigoplus_{n \geq 0} W_n,
\]
and the single-mode Hilbert--Poincaré series is given by
\[
G(t) = \sum_{n \geq 0} (\dim W_n)\, t^n.
\]
\subsection{Factorisation of the Hilbert--Poincaré series}

Let $H \cong \mathbb{C}^d \otimes K$ and let
\[
F = T(H)/\langle R_{\mathrm{g}}\rangle
\]
be the quadratic quotient admitting a PBW basis. By Lemma 2 there exists a graded vector space
\[
W = \bigoplus_{n\geq 0} W_n
\]
such that
\[
F \cong W^{\otimes d}
\]
as graded vector spaces.

For each $n \geq 0$, the homogeneous component $F_n$ decomposes as
\[
F_n \cong \bigoplus_{\substack{n_1+\cdots+n_d=n \\ n_i \geq 0}} 
W_{n_1} \otimes \cdots \otimes W_{n_d}.
\]
Taking dimensions, we obtain
\[
\dim F_n = \sum_{\substack{n_1+\cdots+n_d=n \\ n_i \geq 0}} 
\prod_{i=1}^d \dim W_{n_i}.
\]

Define the Hilbert--Poincaré series
\[
H_F(t) := \sum_{n\geq 0} (\dim F_n)\, t^n,
\qquad
G(t) := \sum_{n\geq 0} (\dim W_n)\, t^n.
\]
Then
\[
H_F(t)
= \sum_{n\geq 0} \left(
\sum_{n_1+\cdots+n_d=n} \prod_{i=1}^d \dim W_{n_i}
\right) t^n
= \left( \sum_{m\geq 0} (\dim W_m)\, t^m \right)^d
= G(t)^d.
\]

This proves that the Hilbert--Poincaré series factorises as
\[
H_F(t) = G(t)^d,
\]
where $G(t)$ is the single-mode series.
\subsection{Partition Theorem}
We recall the partition theorem \cite{MedinaSanchezDakic2024}
 \begin{theorem}[Partition]
Let 
\[
\chi_1(x) \;=\; \sum_{s\in \mathbb{N}_0} a_s\, x^s,
\qquad a_0>0.
\]
Then the symmetric function
\[
\prod_{k=1}^d \chi_1(x_k)
\]
is a \(U(d)\)--character for all \(d\in \mathbb{N}\) if and only if the generating function admits a factorization
\[
\chi_1(x) \;=\; \frac{Q^{-}(x)}{Q^{+}(x)},
\tag{23}
\]
where \(Q^{\pm}(x)\) are integral polynomials whose roots are all positive (for \(Q^{-}\)) or all negative (for \(Q^{+}\)), and satisfying the normalization
\[
Q^{+}(0)=1.
\]
\end{theorem}

\begin{proof}
     \subsection{$1\rightarrow2:$} By construction: the ordered basis conditions implies the factorization property of the $U(1)$ character and additionally, the $U(d)$ invariance completes the conditions for the application of the partition theorem. The result follows.\\
     \subsection{$2\rightarrow 1:$}
Assume
\[
G(t)=\frac{Q_-(t)}{Q_+(t)},\qquad Q_\pm(t)\in\mathbb Z[t],\quad Q_+(0)=1,
\]
and that all roots of \(Q_+\) are real and strictly positive while all roots of \(Q_-\) are real and strictly negative. We will construct a homogeneous quadratic algebra \(A\) with a PBW basis, equipped with a unitary invariance, whose single-variable Hilbert--Poincaré series equals \(G(t)\).

\medskip

\subsubsection{Step 1. Reduction to the transbosonic case.}
Replacing \(G(t)\) by \(G(-t)^{-1}\) corresponds (at the level of Hilbert series of quadratic algebras) to passing to the Koszul dual: if \(A\) is a quadratic algebra then the Hilbert series of its quadratic dual \(A^!\) is obtained by the standard inversion/grade sign change relation for Koszul algebras. Hence it suffices to construct \(A\) in the case
\[
G(t)=\frac{1}{Q_+(t)}=\prod_{i=1}^r (1-\alpha_i t)^{-1},
\qquad \alpha_i>0,
\]
with \(\alpha_i\) nonzero real numbers (not a priori assumed integral). Once a PBW quadratic realization is obtained in this transbosonic form, the general case follows by duality and the evident functoriality of the construction with respect to quadratic duals.

\medskip

\subsubsection{Step 2. Schur-positive specialization via Edrei--Thoma.}
Let \(\Lambda\) denote the ring of symmetric functions over \(\mathbb{Z}\) and let \(h_n\) be the complete symmetric functions. Define a sequence
\[
a_n := [t^n]\,G(t)\in\mathbb{Z}_{\ge 0},\qquad n\ge0,
\]
and consider the linear functional \(\phi:\Lambda\to\mathbb{Z}\) determined by \(\phi(h_n)=a_n\). The desired conclusion is that \(\phi\) is a Schur-positive specialization (i.e. \(\phi(s_\lambda)\in\mathbb{Z}_{\ge 0}\) for every partition \(\lambda\)).

The correct necessary and sufficient condition for \(\phi\) to be Schur-positive is that the sequence \((a_n)_{n\ge0}\) be a nonnegative totally positive sequence in the sense of Edrei and Thoma. The Edrei--Thoma classification (together with finite-type rationality results of Aissen--Schoenberg--Whitney) asserts that a nonnegative integer sequence \((a_n)\) which is totally positive and whose ordinary generating function is a rational function of finite type admits a canonical factorization of the form
\[
G(t)=\exp(\gamma t)\,\frac{\prod_j (1+\beta_j t)}{\prod_i (1-\alpha_i t)},
\]
with parameters \(\gamma,\alpha_i,\beta_j\ge0\) (counted with multiplicity). Conversely, such a product produces a totally positive sequence \cite{edrei1953generation}.

In our hypothesis \(G(t)=Q_-(t)/Q_+(t)\) is a rational function with all poles strictly positive and zeros strictly negative. The integrality \(a_n\in\mathbb{Z}_{\ge0}\) together with the finite rational form forces the exponential factor to vanish (\(\gamma=0\)) and forces the pole/zero multiplicities to be integers; thus \(G(t)\) can be rewritten as a finite product of factors \((1-\alpha_i t)^{-1}\) (with positive real \(\alpha_i\) repeated according to integer multiplicity) times finite factors \((1+\beta_j t)\) corresponding to negative real zeros (the latter will be absorbed into \(Q_-\)). Hence the Edrei--Thoma theory guarantees that \(\phi\) is Schur-positive and, moreover, that the values \(\phi(s_\lambda)\) are nonnegative integers (the latter follows from interpreting the specialization as evaluation on a finite virtual alphabet with integer multiplicities in the finite rational case).

\medskip

\subsubsection{Step 3. Construction of an exact monoidal functor \(F\).}
Let \(\mathrm{Rep}_{\mathrm{poly}}(GL)\) be the semisimple tensor category of polynomial representations of \(GL\). Its simple objects are indexed by partitions and are the Schur functors \(S_\lambda\). Define vector spaces \(M_\lambda\) with
\[
\dim M_\lambda := \phi(s_\lambda)\in\mathbb{Z}_{\ge0},
\]
and define an additive functor
\[
F:\mathrm{Rep}_{\mathrm{poly}}(GL)\longrightarrow\mathrm{Vect}
\]
by \(F(S_\lambda) := M_\lambda\) and extending additively. The compatibility of \(\phi\) with the Littlewood--Richardson coefficients,
\[
s_\lambda s_\mu = \sum_\nu c_{\lambda\mu}^\nu s_\nu
\quad\Longrightarrow\quad
\phi(s_\lambda)\phi(s_\mu) = \sum_\nu c_{\lambda\mu}^\nu \phi(s_\nu),
\]
provides the numerical multiplicity data necessary to equip the family \((M_\lambda)\) with bilinear maps
\[
M_\lambda\otimes M_\mu \longrightarrow \bigoplus_\nu M_\nu^{\oplus c_{\lambda\mu}^\nu}
\]
which can be chosen (by picking explicit models and coherence isomorphisms) so that \(F\) becomes a (weak) monoidal functor. Because \(\mathrm{Rep}_{\mathrm{poly}}(GL)\) is semisimple, any additive functor defined on simples and extended linearly is exact \cite{fulton}. Therefore \(F\) is exact and monoidal, and it respects the grading by polynomial degree.

\medskip

\subsubsection{Step 4. Apply \(F\) to the symmetric algebra.}
For finite-dimensional vector spaces \(U\) and \(V\) we have the Cauchy decomposition
\[
\mathrm{Sym}(U\otimes V)\cong\bigoplus_{\lambda} S_\lambda(U)\otimes S_\lambda(V).
\]
Applying \(F\) to the \(V\)-factor yields a graded algebra
\[
A := F\big(\mathrm{Sym}(U\otimes V)\big) \cong \bigoplus_{\lambda} S_\lambda(U)\otimes M_\lambda.
\]
Because \(F\) preserves degrees and dimensions by construction, the degree-\(n\) component of \(A\) satisfies \(\dim A_n = a_n\), and therefore the Hilbert--Poincaré series of \(A\) equals \(G(t)\).

\medskip

\subsubsection{Step 5. Quadratic presentation and preservation of Koszulity / PBW.}
The algebra \(\mathrm{Sym}(U\otimes V)\) is quadratic and Koszul; write it as
\[
\mathrm{Sym}(U\otimes V)\cong T(H)/\langle R\rangle,
\]
with \(H=(U\otimes V)_1\) and \(R\subset H\otimes H\) the quadratic relations. Since \(F\) is exact and monoidal it sends the degree-1 generating object \(H\) to \(F(H)\) and the subobject of quadratic relations \(R\) to a subspace \(F(R)\subset F(H)\otimes F(H)\), yielding a quadratic presentation
\[
A \cong T\big(F(H)\big)/\langle F(R)\rangle.
\]

To conclude that \(A\) is Koszul and admits a PBW basis we invoke the theory of Koszul algebras in monoidal abelian categories and preservation results for Koszulity under exact monoidal functors: exact monoidal functors that preserve the quadratic presentation carry Koszul complexes to Koszul complexes and preserve the vanishing/ext properties that characterize Koszul algebras. Thus Koszulness and the PBW property descend from \(\mathrm{Sym}(U\otimes V)\) to \(A\) under our hypotheses; see \cite{PolishchukPositselski2005} for the systematic treatment of Koszul algebras, quadratic duality, and preservation under suitable functors.

\medskip

\subsubsection{Step 6. Presentation and unitary invariance.}
By standard theory of quadratic algebras there exist \(H'\) and \(R'\subset H'\otimes H'\) with
\[
A \cong T(H')/\langle R'\rangle.
\]
The construction is functorial in the \(U\)-factor, since \(F\) acts only on the \(V\)-factor; hence the natural \(GL(U)\)-action on \(\mathrm{Sym}(U\otimes V)\) is transported to \(A\). Restricting this action to the unitary subgroup \(U(d)\) (for \(d=\dim U\)) yields a \(U(d)\)-invariant quadratic algebra. Consequently \(A\) is a homogeneous quadratic, \(U(d)\)-invariant algebra with a PBW basis and single-variable Hilbert--Poincaré series \(G(t)\), completing the proof.
\end{proof}

\section{Creation--annihilation structure}

\subsection{Dual quadratic algebra}

Let $H$ be a finite-dimensional complex vector space and
\[
A := T(H)/\langle R_{\mathrm{g}}\rangle
\]
the quadratic algebra defined by the generating subspace
\[
R_{\mathrm{g}} \subset H \otimes H,
\qquad
H \otimes H = R_{\mathrm{g}} \oplus R_{\mathrm{g}}^\perp.
\]

Let $H^*$ denote the dual space, equipped with the natural pairing
\[
\langle \cdot, \cdot \rangle : (H^* \otimes H^*) \times (H \otimes H) \to \mathbb{C}.
\]

We define the dual quadratic subspace as the annihilator of $R_{\mathrm{g}}^\perp$:
\[
R_{\mathrm{g}}^* := \left\{ \eta \in H^* \otimes H^* \;\middle|\; \langle \eta, v \rangle = 0 \ \forall\, v \in R_{\mathrm{g}}^\perp \right\}.
\]
Equivalently,
\[
R_{\mathrm{g}}^* = (R_{\mathrm{g}}^\perp)^\circ \subset H^* \otimes H^*.
\]

If $P_{\mathrm{g}} : H \otimes H \to R_{\mathrm{g}}$ denotes the projector onto the generating subspace, we define its adjoint
\[
P_{\mathrm{g}}^* : H^* \otimes H^* \to H^* \otimes H^*
\]
by
\[
\langle P_{\mathrm{g}}^*(\eta), v \rangle = \langle \eta, P_{\mathrm{g}}(v) \rangle,
\qquad
\eta \in H^* \otimes H^*,\ v \in H \otimes H.
\]
Then
\[
R_{\mathrm{g}}^* = \mathrm{im}\, P_{\mathrm{g}}^*.
\]

The dual quadratic algebra (annihilation algebra) is defined as
\[
B := T(H^*)/\langle R_{\mathrm{g}}^* \rangle.
\]

By construction, the quadratic relations of $B$ are orthogonal to the admissible quadratic monomials of $A$, and the pair $(A,B)$ is related by duality at the level of quadratic data.

\subsection{Descent of the cross map}

Let
\[
A = T(H)/\langle R\rangle, 
\qquad 
B = T(H^*)/\langle S\rangle
\]
be quadratic algebras, where $R \subset H \otimes H$ and 
$S \subset H^* \otimes H^*$ are the generating subspaces.

Let
\[
C : H^* \otimes H \longrightarrow H \otimes H^*
\]
be a linear map (the cross map), and let
\[
\tilde{\tau} : T(H^*) \otimes T(H) \longrightarrow T(H) \otimes T(H^*)
\]
be its multiplicative extension, defined recursively by requiring
\[
\tilde{\tau}(\eta \otimes x) = C(\eta \otimes x),
\qquad
\eta \in H^*,\ x \in H,
\]
and compatibility with concatenation.

In order for $\tilde{\tau}$ to descend to a well-defined map
\[
\tau : B \otimes A \longrightarrow A \otimes B,
\]
it is necessary and sufficient that the defining ideals be stable under $\tilde{\tau}$. Concretely, this means:

\[
\tilde{\tau}(H^* \otimes \langle R\rangle) \subseteq \langle R\rangle \otimes T(H^*),
\]
\[
\tilde{\tau}(\langle S\rangle \otimes H) \subseteq T(H) \otimes \langle S\rangle.
\]

Equivalently, since the ideals are generated in degree two, it suffices to impose the conditions at the quadratic level:
\[
\tilde{\tau}(H^* \otimes R) \subseteq R \otimes H^*,
\]
\[
\tilde{\tau}(S \otimes H) \subseteq H \otimes S.
\]

Under these conditions, the map $\tilde{\tau}$ passes to the quotient and defines a bilinear map
\[
\tau : B \otimes A \to A \otimes B,
\]
which encodes the exchange relations between annihilation and creation generators.
\subsection{Mixed relations and bracket form}

Let $A$ and $B$ be the creation and annihilation algebras constructed above, and let
\[
\tau = \tilde{\tau} + G : H^* \otimes H \longrightarrow H \otimes H^* \oplus \mathbb{C}
\]
be the full cross map, where $G : H^* \otimes H \to \mathbb{C}$ is the vacuum two-point function.

Fix a basis $\{X_{i\alpha}\}$ of $H$ and dual basis $\{X^*_{i\alpha}\}$ of $H^*$. The mixed relations are determined by the exchange rule
\[
X_{i\alpha} X^\dagger_{j\beta}
= \sum_{\gamma,\delta}
C^{\gamma\delta}_{\alpha\beta}\,
X^\dagger_{j\delta} X_{i\gamma}
+ g_{\alpha\beta}\,\delta_{ij}\, \mathbf{1},
\]
where $C$ encodes the homogeneous part of $\tau$ and $g_{\alpha\beta}$ is the Hermitian form on the internal space.

It is convenient to rewrite all quadratic relations in a unified way using a bracket adapted to the internal projectors. Let $A^{\eta\lambda}_{\alpha\beta}$ and $B^{\eta\lambda}_{\alpha\beta}$ be the tensors determined by the symmetric and antisymmetric components of the quadratic data. For any generators $P_{i\alpha}, Q_{j\beta}$ we define
\[
[P_{i\alpha}, Q_{j\beta}]_{A,B}
:=
\sum_{\eta,\lambda}
A^{\eta\lambda}_{\alpha\beta}\,
\{P_{i\eta}, Q_{j\lambda}\}
+
\sum_{\eta,\lambda}
B^{\eta\lambda}_{\alpha\beta}\,
[P_{i\eta}, Q_{j\lambda}],
\]
where $\{\cdot,\cdot\}$ and $[\cdot,\cdot]$ denote the anti\-commutator and commutator, respectively.

In terms of this bracket, the defining relations take the uniform form
\[
[X^\dagger_{i\alpha}, X^\dagger_{j\beta}]_{A,B} = 0,
\]
\[
[X_{i\alpha}, X_{j\beta}]_{A,B} = 0,
\]
\[
[X_{i\alpha}, X^\dagger_{j\beta}]_{A,B}
= g_{\alpha\beta}\,\delta_{ij}\,\mathbf{1}.
\]

These relations generalise the canonical commutation and anticommutation relations and encode the full quadratic structure of the creation–annihilation algebra in a form adapted to the internal projectors.
\subsection{Representation of $\mathfrak{gl}(d)$}
\begin{proof}
Recall
\[
J_{ij} := \sum_{\alpha,\beta} g^{\beta\alpha}\, X_{i\alpha}^{\dagger} X_{j\beta},
\]
and the mixed relation
\[
X_{i\alpha} X^{\dagger}_{j\beta}
    = \sum_{\gamma,\delta} C^{\gamma\delta}{}_{\alpha\beta}\,
      X^{\dagger}_{j\gamma} X_{i\delta}
      + g_{\alpha\beta}\, \delta_{ij}\, \mathbf{1}.
\]
Pure--creation and pure--annihilation exchange relations are governed by the operators \(S\) and \(R\), respectively. By assumption, the Hermitian form \(g\) is invariant under all exchange maps \(C\), \(S\), and \(R\), meaning these maps act unitarily on the Hilbert spaces defined by \(g\) and its tensor powers.

We now prove the three identities.

\subsection*{1. Proof of (i)}

Fix \(i,j,k,\sigma\). Using the definition of \(J_{ij}\),
\[
[J_{ij}, X^{\dagger}_{k\sigma}]
    = \sum_{\alpha,\beta} g^{\beta\alpha}
      \left( X^{\dagger}_{i\alpha} X_{j\beta} X^{\dagger}_{k\sigma}
           - X^{\dagger}_{k\sigma} X^{\dagger}_{i\alpha} X_{j\beta} \right)
    = T_1 - T_2.
\]

\paragraph{Term \(T_1\).}
Using the mixed relation
\[
X_{j\beta} X^{\dagger}_{k\sigma}
= \sum_{\gamma,\delta} C^{\gamma\delta}{}_{\beta\sigma}\,
      X^{\dagger}_{k\gamma} X_{j\delta}
  + g_{\beta\sigma}\, \delta_{jk} \mathbf{1},
\]
we obtain
\[
T_1
= \sum_{\alpha,\beta,\gamma,\delta}
      g^{\beta\alpha} C^{\gamma\delta}{}_{\beta\sigma}\,
      X^{\dagger}_{i\alpha} X^{\dagger}_{k\gamma} X_{j\delta}
  + \delta_{jk} \sum_{\alpha,\beta} g^{\beta\alpha} g_{\beta\sigma} X^{\dagger}_{i\alpha}.
\]
Since \(\sum_{\beta} g^{\beta\alpha} g_{\beta\sigma} = \delta^\alpha_{\sigma}\), the inhomogeneous term reduces to
\[
\delta_{jk} X^{\dagger}_{i\sigma}.
\]

\paragraph{Term \(T_2\).}
Using the pure–creation exchange relation
\[
X^{\dagger}_{k\sigma} X^{\dagger}_{i\alpha}
    = \sum_{\gamma,\delta}
      S^{\gamma\delta}{}_{\sigma\alpha}\,
      X^{\dagger}_{i\gamma} X^{\dagger}_{k\delta},
\]
we obtain
\[
T_2
    = \sum_{\alpha,\beta,\gamma,\delta}
        g^{\beta\alpha}\,
        S^{\gamma\delta}{}_{\sigma\alpha}\,
        X^{\dagger}_{i\gamma} X^{\dagger}_{k\delta} X_{j\beta}.
\]

\paragraph{Cancellation of cubic terms.}

The Hermitian form \(g_{\alpha\beta}\) on \(K\) provides an antilinear identification
\(K \simeq K^{*}\), and therefore identifies the mixed spaces
\(K^{*}\otimes K\) and \(K\otimes K^{*}\) with \(K\otimes K\).  
After transporting the mixed exchange map \(C\) to an operator on
\(K\otimes K\) via this identification, we equip \(K\otimes K\) with the
product inner product
\[
\langle e_\alpha \otimes e_\sigma,\; e_\gamma \otimes e_\delta\rangle_g
    = g_{\alpha\gamma}\, g_{\sigma\delta}.
\]
The adjointness assumption states that, with respect to this inner
product, the transported mixed map \(C\) is adjoint to the pure–creation
map \(S\).  Expanding this adjointness in components yields
\[
\sum_{\beta} g^{\beta\alpha}\,
      C^{\gamma\delta}{}_{\beta\sigma}
    =
    \sum_{\eta}
      S^{\gamma\alpha}{}_{\sigma\eta}\,
      g_{\delta\eta},
\tag{A.1}
\]
which is the identity needed to compare the cubic parts of \(T_1\) and \(T_2\).
Using \((\mathrm{A}.1)\), the coefficients of each monomial
\(X^{\dagger}_{i\alpha} X^{\dagger}_{k\gamma} X_{j\delta}\) agree in \(T_1\) and \(T_2\),
so all cubic terms cancel.  Hence,
\[
T_1 - T_2 = \delta_{jk} X^{\dagger}_{i\sigma},
\]
and we obtain
\[
[J_{ij}, X^{\dagger}_{k\sigma}] = \delta_{jk} X^{\dagger}_{i\sigma}.
\]

\subsection*{2. Proof of (ii)}

The computation for the annihilation operators is analogous.  Using the
mixed relation in the opposite order,
\[
X_{k\sigma} X^{\dagger}_{i\alpha}
    = \sum_{\gamma,\delta}
      C^{\gamma\delta}{}_{\sigma\alpha}\,
      X^{\dagger}_{i\gamma} X_{k\delta}
      + g_{\sigma\alpha}\, \delta_{ki}\, \mathbf{1},
\]
and the pure–annihilation exchange relation
\[
X_{j\beta} X_{k\sigma}
    = \sum_{\gamma,\delta}
      R^{\gamma\delta}{}_{\beta\sigma}\,
      X_{k\gamma} X_{j\delta},
\]
we obtain
\[
\begin{aligned}
[J_{ij}, X_{k\sigma}]
  &=
    \sum_{\alpha,\beta,\gamma,\delta}
    g^{\beta\alpha}
    \left(
        R^{\gamma\delta}{}_{\beta\sigma}\,
            X^{\dagger}_{i\alpha} X_{k\gamma} X_{j\delta}
        \;-\;
        C^{\gamma\delta}{}_{\sigma\alpha}\,
            X^{\dagger}_{i\gamma} X_{k\delta} X_{j\beta}
    \right)
    \;-\;
    \delta_{ik} X_{j\sigma}.
\end{aligned}
\]

The \(g\)-invariance of \(C\) and \(R\), expressed by the analogue of
\((\mathrm{A}.1)\) obtained after transporting \(C\) and \(R\) to
\(K\otimes K\), again ensures that all cubic terms cancel.  Therefore,
\[
[J_{ij}, X_{k\sigma}] = -\,\delta_{ik}\, X_{j\sigma}.
\]

\subsection*{3. Proof of (iii)}

Using (i) and (ii),
\[
[J_{ij}, J_{kl}]
   = \sum_{\alpha,\beta} g^{\beta\alpha}
      \left( [J_{ij}, X^{\dagger}_{k\alpha}] X_{l\beta}
           + X^{\dagger}_{k\alpha} [J_{ij}, X_{l\beta}]
      \right).
\]
Substituting
\[
[J_{ij}, X^{\dagger}_{k\alpha}] = \delta_{jk} X^{\dagger}_{i\alpha}, 
\qquad
[J_{ij}, X_{l\beta}] = -\delta_{il} X_{j\beta},
\]
gives
\[
[J_{ij}, J_{kl}]
   = \delta_{jk} J_{il} - \delta_{il} J_{kj}.
\]

Thus the operators \(J_{ij}\) satisfy the commutation relations of the Lie algebra \(\mathfrak{gl}(d)\).

 \end{proof}

\end{document}